\documentclass[sigconf,9pt,nonacm]{acmart}

\usepackage{algorithm}
\usepackage[noend]{algpseudocode}
\usepackage{graphicx}
\usepackage{adjustbox}
\usepackage{caption}
\usepackage{textcomp}
\usepackage{amsthm}
\usepackage{verbatim}
\usepackage{subcaption}
\usepackage{multirow}
\usepackage{booktabs}
\usepackage{makecell}

\usepackage{booktabs}
\usepackage[table]{xcolor}
\usepackage{array}
\usepackage{nicematrix}
\usepackage{tikz}
\usepackage{threeparttable}
\usepackage{tablefootnote}
\usepackage{bm}
\usepackage{mathtools}
\usepackage{algorithm}
\usepackage[noend]{algpseudocode}
\usepackage{enumitem}
\usepackage{dblfloatfix}

\def\BibTeX{{\rm B\kern-.05em{\sc i\kern-.025em b}\kern-.08em
    T\kern-.1667em\lower.7ex\hbox{E}\kern-.125emX}}

\newcommand{\WE}{\mathcal{W}_E}
\newcommand{\WV}{\mathcal{W}_V}

\newtheorem{problem}{\textbf{Problem}}

\newtheorem{theorem}{\textbf{Theorem}}
\DeclareMathOperator*{\argmax}{arg\,max}
\DeclareMathOperator{\WaterFill}{WaterFilling}

\newcommand*\circled[1]{\tikz[baseline=(char.base)]{
    \node[shape=circle, fill=black, text=white, inner sep=0.8pt, font=\sffamily\bfseries\small] (char) {#1};}}

\definecolor{ForestGreen}{rgb}{0.13, 0.55, 0.13}

\copyrightyear{2026}
\acmYear{2026}
\setcopyright{cc}
\setcctype{by}
\acmConference[DAC '26]{63rd ACM/IEEE Design Automation Conference}{July 26--29, 2026}{Long Beach, CA, USA}
\acmBooktitle{63rd ACM/IEEE Design Automation Conference (DAC '26), July 26--29, 2026, Long Beach, CA, USA}
\acmDOI{10.1145/3770743.3804131}
\acmISBN{979-8-4007-2254-7/2026/07}

\begin{document}

\fancyhead{}
\renewcommand{\headrulewidth}{0pt}

\title{ComPart: Community-Guided Post-Coarsening for High-Quality Hypergraph Partitioning}

\author{Yugao Zhu,\ Zhicheng Guo, Yuchao Wu,  Mengming Li, Jing Wang, Zhiyao Xie}
\authornote{corresponding author}
\affiliation{%
  \institution{Hong Kong University of Science and Technology}
  \country{\{yzhuel, zguobx, ywu092, mengming.li, jwangjw\}@connect.ust.hk, eezhiyao@ust.hk}
}

\begin{abstract}    
    \label{sec:abs}
    Hypergraph partitioning is a critical step in the design of complex embedded systems, essential for optimizing task mapping on heterogeneous MPSoCs and enabling multi-FPGA prototyping. 
Many existing methods rely on community detection to identify modules with dense internal and sparse external connections, typically utilizing them to constrain the coarsening phase—a widely adopted paradigm. 
In this work, we propose ComPart, a generalized framework that integrates diverse community detection methods to uncover high-quality clusterings throughout the post-coarsening stages (i.e., initial partitioning and uncoarsening).
These discovered clusterings serve as distinct structural guides, enabling the refinement process to identify superior partitioning solutions. 
Our framework offers two key advantages: (1) it establishes a new paradigm that leverages community structures detected during uncoarsening to escape local optima and explore globally meaningful solution subspaces, transcending the limitations of standard local refinements; and (2) it flexibly accommodates both existing and future community detection methods. 
Furthermore, we theoretically generalize locally-dense decomposition—originally from graphs—to the hypergraph domain. 
We provide the formal extension and necessary proofs to apply this technique to hypergraphs, marking its first application in hypergraph partitioning. 
Specifically, we utilize this rigorously derived decomposition to guide the initial partitioning phase toward superior starting points. 
Experimental results on standard benchmarks demonstrate that our method consistently outperforms state-of-the-art methods in solution quality.
\end{abstract}

\maketitle

\section{Introduction}
    \label{sec:intro}
    Hypergraph partitioning serves as a fundamental algorithmic kernel in the design and verification of modern embedded systems. 
It is essential not only for task mapping on heterogeneous MPSoCs during hardware/software co-design\cite{wolf2008multiprocessor,dick1998tgff}, but also for multi-FPGA prototyping~\cite{li2024mapart,tong2024easypart} in logic emulation flows. 
In both scenarios, modeling system components and dependencies as hypergraphs allows for minimizing inter-device communication and satisfying strict resource constraints. 
Over the past decades, a wide range of methods have been developed to address these combinatorial challenges, including spectral approaches~\cite{aghdaei2022hyperef,bustany2022specpart,bustany2023k}, max-flow/min-cut formulations~\cite{gottesburen2019evaluation}, deep learning-based methods~\cite{liang2024medpart,chen2024hypergraph}, and multi-level frameworks.\looseness=-1

Among these methodologies, the multi-level framework has established itself as the dominant paradigm for high-performance hypergraph partitioning, underpinning leading solvers like PaToH~\cite{ccatalyurek2011patoh}, hMETIS~\cite{karypis1997multilevel}, and KaHyPar~\cite{schlag2023high}. This framework typically consists of three key phases:
(1) \textbf{Coarsening}: densely connected vertices are iteratively merged to construct a hierarchy of progressively smaller hypergraphs until the problem size allows for direct partitioning;
(2) \textbf{Initial partitioning}: a initial solution is computed on the coarsest hypergraph;
(3) \textbf{Uncoarsening and refinement}: the solution is projected back to finer levels, where local heuristics such as Fiduccia-Mattheyses (FM)~\cite{fiduccia1988linear} or Kernighan-Lin (KL)~\cite{kernighan1970efficient} are applied to iteratively improve solution quality.

\begin{figure}[!t]
    \centering
    \includegraphics[width=\linewidth]{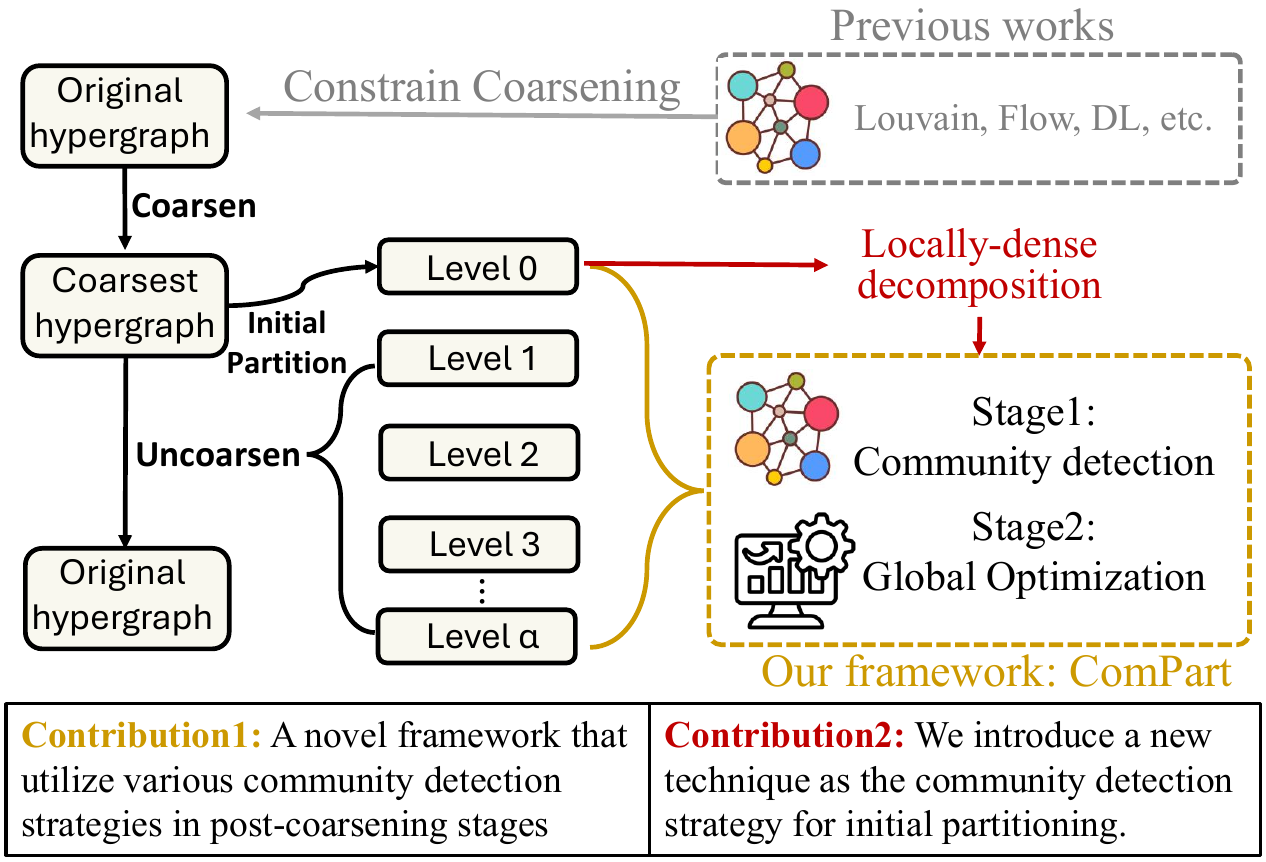}
    \vspace{-.2in}
    \caption{Framework Comparison between our method ComPart and prior works. }
    \label{fig:overview}
    \vspace{-.2in}
\end{figure}

Multi-level frameworks have undergone continuous enhancement, with one of the most prominent directions being the integration of \textit{community detection} strategies. These methods typically introduce community detection \textbf{for coarsening}, where the coarsening process is restricted within each detected community. For example, KaHyPar adopts the Louvain method~\cite{blondel2008fast} to guide the coarsening process, while SHyPar~\cite{sajadinia2025shypar} employs a flow-based community detection approach to optimize hypergraph local conductance (HLS). Both the V-cycle~\cite{karypis1997multilevel} and recombination methods~\cite{andre2018memetic,acikalin2022multilevel,popp2021multilevel} treat candidate partitions as distinct communities, restricting coarsening within their boundaries. Similarly, GenPart~\cite{chen2024hypergraph} utilizes a graph convolutional network (GCN)-based generative model to produce initial solutions before executing the V-cycle.

As illustrated in Fig. \ref{fig:overview}, while almost all prior works concentrate on utilizing community detection to guide the coarsening process, they largely neglect its potential in other stages. In this work, we unprecedentedly target these neglected stages. To this end, we propose \textbf{ComPart}, a novel framework that leverages various community detection methods throughout the \textbf{post-coarsening phases} (i.e., initial partitioning and uncoarsening). This design enables the integration of a wide range of existing community-based techniques—spectral, modularity-based, deep learning, and beyond—thereby generalizing previous coarsening-centric approaches and introducing a new, flexible paradigm centered on the post-coarsening process.

Specifically, this paradigm allows us to transcend the limitations of traditional refinement. Instead of relying solely on vertex-wise moves like FM refinement—which often yields only marginal improvements and becomes trapped in local optima—we introduce a \textit{structure-aware} optimization strategy. We \textbf{periodically} identify community structures within the post-coarsening phases and formulate their reassignment as an integer linear programming (ILP) problem~\cite{heuer2015engineering}. This allows us to optimize the movement of entire community fragments simultaneously. By shifting from a local, node-centric view to a global, community-centric perspective, our approach facilitates large-scale adjustments, effectively escaping the local optima that limit existing frameworks.

In addition, we introduce a theoretically grounded innovation at the coarsest level by generalizing \textit{locally-dense decomposition}~\cite{danisch2017large,zhu2023fast}—originally a graph-theoretic technique—to the hypergraph domain. We provide the formal extension and necessary proofs to adapt this concept to hypergraphs, marking its first application in hypergraph partitioning. 
Crucially, unlike heuristic clustering methods, this decomposition is formulated as a \textbf{convex optimization problem}, guaranteeing convergence to a theoretically optimal solution that reveals hierarchically nested dense substructures. 
Leveraging these mathematically rigorous structural insights, we guide the initial partitioning toward superior starting points.\looseness=-1

The main contributions of this work are summarized as follows: 

\begin{enumerate}
    \item We propose a \emph{generalized framework}, ComPart, that integrates diverse community detection methods into the post-coarsening stages to uncover high-quality clustering structures. Unlike traditional approaches that rely solely on local refinements, our framework leverages these identified communities to guide global solution adjustments, allowing the algorithm to escape local optima. This establishes a novel paradigm that transcends the limitations of local optimization, providing new insights into the effective utilization of community detection in hypergraph partitioning.
    
    \item We generalize the \emph{locally-dense decomposition} technique from the graph to the hypergraph domain. We provide the theoretical extension required to adapt this rigorous method to hyperedges, marking its first application in hypergraph partitioning. By identifying hierarchically nested dense structures within the coarsest hypergraph, this decomposition serves as a principled guide for generating superior initial solutions.
    
    \item We conduct extensive experiments on standard benchmark datasets, comparing our method against both mainstream solvers and recent state-of-the-art approaches. The results demonstrate that our algorithm consistently achieves superior solution quality. Furthermore, comprehensive ablation studies validate the efficacy of each component in our framework. Additionally, we perform a comparison of different community detection strategies to evaluate their specific impact on partitioning performance.
\end{enumerate}

\section{Problem Formulation}
    \label{sec:pro}
    Given a hypergraph $H(V,E)$, where $V$ and $E$ denote the sets of vertices and hyperedges, respectively, each vertex $v \in V$ and hyperedge $e \in E$ is associated with a positive weight $w_v$ and $w_e$. The hypergraph partitioning problem\footnote{Our formulation follows that of established tools such as KaHyPar, PaToH, and the recent hMETIS release (v2.0-pre1).} aims to determine a partition scheme $S$ that divides $V$ into $k$ disjoint blocks $V_1, V_2, \ldots, V_k$, satisfying $V_i \cap V_j = \emptyset$ for all $i \neq j$ and ${\textstyle \bigcup_{i=0}^{k-1}} V_i = V$.
Each block $V_i$ must satisfy the balance constraint
\begin{equation*}
W_{V_i} \le (1+\epsilon)\cdot\left\lceil \frac{W_V}{k} \right\rceil,
\end{equation*}
where $W_V = \sum_{v \in V} w_v$, $W_{V_i} = \sum_{v \in V_i} w_v$, and $\epsilon$ is the imbalance factor.
Under this constraint, the objective is to minimize the cut size of the partition:
\begin{equation*}
    \text{cutsize}_H(S)=\sum_{e\not\subseteq V_i \text{ for all } V_i\in S}{w_e}
\end{equation*}

\section{Methodology}
    \label{sec:methodology}
    In this section, we first present our new framework, ComPart, in which multiple community detection algorithms are integrated into the post-coarsening phases, as detailed in Section~\ref{subsec:framework}. 
Subsequently, in Section~\ref{subsec:ldd}, we formally generalize \textit{locally-dense decomposition} from graphs to hypergraphs, extending both its theoretical properties and algorithmic formulation. 
We leverage this decomposition as a specialized community detection strategy to significantly improve the quality of initial partitioning.

\subsection{Our Framework: ComPart}
\label{subsec:framework}

While prior community-aware frameworks predominantly utilize community structure to restrict coarsening, they significantly underutilize its potential in post-coarsening stages. Consequently, ComPart introduces a fundamental paradigm shift: instead of guiding coarsening, we unprecedentedly target the post-coarsening phases---specifically, initial partitioning and uncoarsening. The overall architecture of ComPart, highlighting this novel approach, is depicted in Figure \ref{fig:framework}.

Our framework is built upon the $n$-level architecture of \texttt{KaHyPar}. Here, $n$ denotes the node count of the original hypergraph: each coarsening step merges one pair of nodes, and each uncoarsening step restores one pair while performing local refinement. During the uncoarsening phase, we perform a two-stage operation at $\alpha$ distinct levels. These levels are triggered whenever the number of nodes reaches a value defined by the set:
\[
\left\{ n_c^{1 - 1/\alpha} n^{1/\alpha},\; n_c^{1 - 2/\alpha} n^{2/\alpha},\; \ldots,\; n_c^{1/\alpha} n^{1 - 1/\alpha},\; n \right\},
\]
where $n_c$ is the node count of the coarsest hypergraph. Intuitively, as the node count increases proportionally, the topological details evolve smoothly across levels, which aligns with the ``Kronecker product'' principle~\cite{graham2018kronecker}; thus, applying community detection at these intervals yields more effective results. The two stages are:
\begin{enumerate}
    \item Applying specified community detection algorithms to the hypergraph at this level;
    \item Considering potential moves for \textit{sub-communities}, i.e., the portions of a community created by cuts.
\end{enumerate}

\begin{figure*}[!t]
    \centering
    \includegraphics[width=0.9\linewidth]{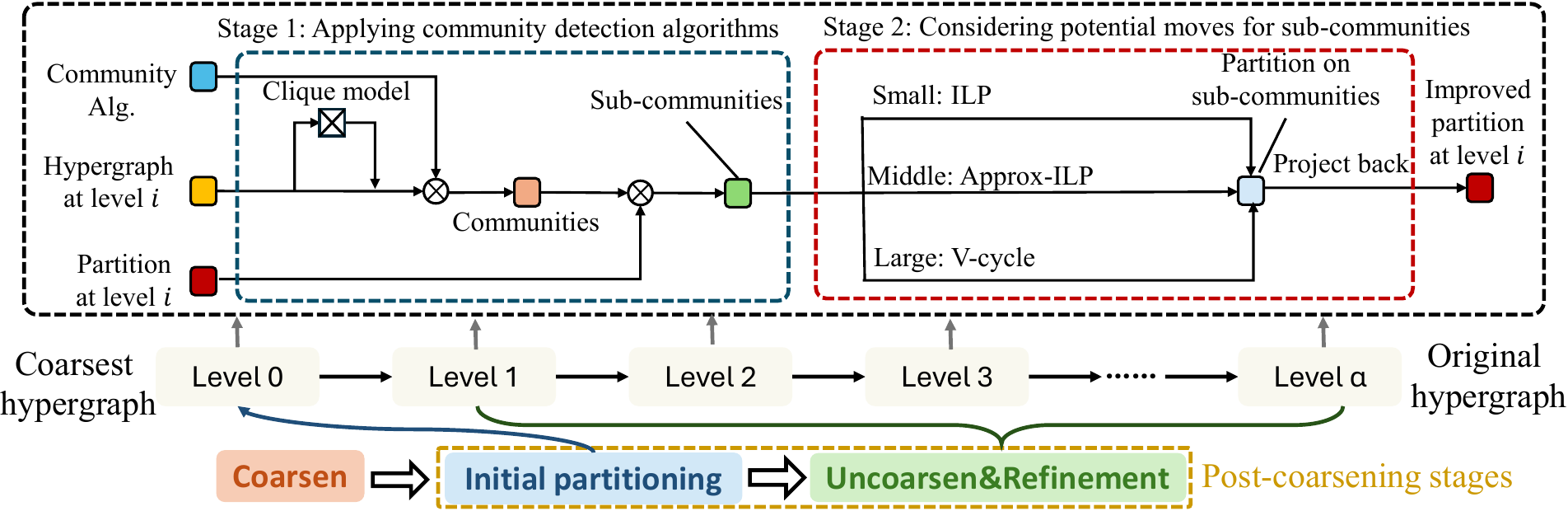}
    \vspace{-.15in}
    \caption{Overview of the proposed ComPart framework. Locally-dense decomposition is used for community detection in the coarsest hypergraph (i.e., level 0).}
    \label{fig:framework}
    \vspace{-.15in}
\end{figure*}

\textbf{Stage \circled{1}: Applying community detection algorithms.}
While some community detection algorithms operate directly on hypergraphs, many are restricted to standard graphs. To leverage these, we convert the hypergraph into a graph using the \textbf{clique model}.\looseness=-1

In this model, each hyperedge $e$ is selectively transformed into a clique by connecting its vertices pairwise with edge weights of $w_e / (|e|-1)$. Specifically, hyperedges with size $n < S_1$ are expanded directly, while those within the range $S_1 \le n < S_2$ are placed into a temporary set $Q$, and larger hyperedges ($n \ge S_2$) are ignored (with $S_1=100$ and $S_2=200$ in our experiments). If the resulting graph is not connected\footnote{Disconnected graphs are incompatible with some methods.}, we iteratively apply clique expansion to the hyperedges stored in the set $Q$, prioritizing those with the smallest number of vertices.
We continue this process until the graph becomes fully connected.
If expanding the hyperedges in $Q$ yields no further reduction in the number of connected components, we introduce edges with minimal weight (e.g., $10^{-6}$) to bridge nodes from different components until the entire graph is connected.
This approach allows us to preserve the structural information of the original hypergraph as much as possible, while ensuring the conversion is completed within a reasonable time complexity.

We apply community detection to the hypergraph or its graph representation using four \texttt{igraph} methods: Label Propagation~\cite{raghavan2007near}, Walktrap~\cite{pons2005computing}, Louvain~\cite{blondel2008fast}, and Fast Greedy~\cite{clauset2004finding}. 
The resulting clusters are categorized as \textit{intra-partition communities} (contained within a single partition) or \textit{edge communities} (spanning multiple partitions). 
Edge communities are intersected with partition boundaries to form smaller segments. 
These segments, together with the intra-partition communities, are treated as candidate \textit{sub-communities} for potential moves in stage 2.

\textbf{Stage \circled{2}: Considering potential moves for sub-communities.}
In this stage, we refine the partition assignment of these sub-communities. Several methods can be employed, such as Integer Linear Programming (ILP), V-cycle, or local heuristic algorithms (e.g., the FM algorithm)\footnote{We employ the first two methods to pursue higher solution quality.}.

The ILP formulation of the hypergraph partitioning problem can be expressed as follows. We define binary decision variables $z_{u,i} \in \{0,1\}$ for each vertex $u$ and partition block $P_i$ ($i \in \{0,\dots,k-1\}$), and $h_{e,i} \in \{0,1\}$ for each hyperedge $e$ and partition block $P_i$. Here, $z_{u,i}=1$ indicates that vertex $u$ is assigned to block $P_i$, while $h_{e,i}=1$ denotes that all vertices of hyperedge $e$ are fully contained within block $P_i$. The constraints are formulated as:

\begin{itemize}
    \item $\displaystyle \sum_{i=0}^{k-1} z_{u,i} = 1$, \quad for all $u \in V$;
    \item $\displaystyle h_{e,i} \leq z_{u,i}$, \quad for all $e \in E$, and for every $u \in e$;
    \item $\displaystyle \sum_{u \in V} w_u z_{u,i} \leq (1+\epsilon)\,\Big\lceil \frac{W_V}{k} \Big\rceil$, \quad for all $i \in \{0,\dots,k-1\}$;
\end{itemize}

The optimization objective is to maximize the total weight of uncut hyperedges:
\[
\text{maximize} \quad \sum_{e \in E} \sum_{i=0}^{k-1} w_e\, h_{e,i}.
\]

Given that ILP inherently exhibits high computational complexity, we adopt a trade-off strategy between runtime and solution quality based on the size of the clustered hypergraph. Specifically, for \textit{small} hypergraphs, we directly solve the ILP using the current partitioning scheme as an initial solution to accelerate convergence; for \textit{medium} hypergraphs, we similarly initialize with the current partitioning but terminate early once the optimality gap drops below 1\%; and for \textit{large} hypergraphs, we substitute the ILP with the V-cycle in KaHyPar to minimize the cut size efficiently.

\subsection{Locally-dense decomposition in initial partitioning}
\label{subsec:ldd}
We introduce the first application of \textit{locally-dense decomposition} to hypergraphs. This adaptation yields a hierarchy of nested, provably denser substructures. Distinct from heuristic approaches, our decomposition is formulated as a \textbf{convex optimization problem with convergence guarantees}, offering a mathematically rigorous foundation for initialization.\looseness=-1

To apply this technique to hypergraph partitioning, we generalize the existing formulation~\cite{zhu2023fast,danisch2017large} in two key aspects: (1) extending the decomposition from \emph{graphs} to \emph{hypergraphs}, and (2) generalizing node weights from \emph{unit weights} to \emph{positive real weights}. We provide rigorous proofs for both generalizations to ensure the theoretical properties are preserved. Given a hypergraph $\mathcal{H}=(\mathcal{V}, \mathcal{E}, \WE, \WV)$, we begin with the following problem definition:\looseness=-1

\begin{problem}[Locally-dense Decomposition Problem on Hypergraph]
    \label{prob:locally-dense}
    Given a hypergraph $\mathcal{H}=(\mathcal{V}, \mathcal{E}, \WE, \WV)$, 
    find its locally-dense decomposition 
    $\emptyset=B_0 \subsetneqq B_1 \subsetneqq \ldots \subsetneqq B_\gamma=\mathcal{V}$, 
    such that $B_i$ is the maximal\footnote{Maximality means that no other set strictly containing $B_{i}$ satisfies the argmax condition.} subgraph strictly containing $B_{i-1}$ that satisfies:
    $ B_i = \argmax_{S \supsetneqq B_{i-1}} \frac{\WE(S) - \WE(B_{i-1})}{\WV(S \setminus B_{i-1})} $
\end{problem}

Here, $\WE(S)$ and $\WE(B_{i-1})$ represent the sum of weights of hyperedges whose endpoints are all contained within the set $S$ and $B_{i-1}$, respectively. $S \setminus B_{i-1}$ denotes the set of vertices in $S$ that are not in $B_{i-1}$, and $\WV(S \setminus B_{i-1})$ represents the sum of weights of these vertices. We introduce a quadratic programming formulation associated with the studied problem. Theorems \ref{prop:uniquity1} and \ref{prop:one-way} characterize their key properties, serving as a basis for Theorem \ref{prop:connection}, which demonstrates the equivalence.

\begin{equation}
    \label{eq:decomposition_hypergraph}
    \centering
    \begin{aligned}
        \textrm{minimize} \qquad  & \sum_{v \in \mathcal{V}}{w_v\ell_v^2} \\
        \textrm{subject to} \qquad & \sum_{v \in e}{f_e(v)} = w_e \qquad \forall e \in \mathcal{E} \\
                                & \ell_v = \frac{\sum_{e \ni v}{f_e(v)}}{w_v}  \qquad \forall v \in \mathcal{V}\\ 
                                & f_e(v) \ge 0 \qquad \forall e \in \mathcal{E}, \forall v \in e \\
    \end{aligned}
\end{equation}

\begin{theorem}
    Given a hypergraph $\mathcal{H}=(\mathcal{V}, \mathcal{E}, \WE, \WV)$, its locally-dense decomposition is unique. 
    \label{prop:uniquity1}
\end{theorem}

The uniqueness is proved via mathematical induction. Specifically, the union of the smallest differing sets $B^1_i$ and $B^2_i$ from two decompositions satisfies the $\operatorname{argmax}$ condition, which contradicts maximality. The detailed derivation is omitted for brevity.

\begin{theorem}
    For an optimal solution $(\bm{f}^*,\bm{\ell}^*)$ and a hyperedge. For any node $u\in{e}$, if there is a node $v\in{e}$ satisfying ${\ell_{u}}^*>{\ell_{v}}^*$, then ${f_e}^*(u)=0$.
    \label{prop:one-way}
\end{theorem}

\begin{proof}[Proof of Theorem 2]
    Define $f^*(u):=\sum_{e \ni u}{f_e(u)}$ and $\ell_u^* = f^*(u) / w_u$.
    Assume $\ell_u^* > \ell_v^*$ and $f_e^*(u) > 0$.
    This implies $w_v f^*(u) > w_u f^*(v)$. Consider shifting $d \in (0, f_e^*(u)]$ from $f_e^*(u)$ to $f_e^*(v)$.
    The new total flows are $f'(u) = f^*(u)-d$ and $f'(v) = f^*(v)+d$.
    We show this shift reduces the objective, which is a contradiction:
    \begin{equation}
    \label{ieq:6}
    \begin{split}
        & w_u(\ell_u')^2 + w_v(\ell_v')^2 < w_u(\ell_u^*)^2 + w_v(\ell_v^*)^2 \\
        % \Leftrightarrow w_u(\frac{f^*(u)-d}{w_u})^2+w_v(\frac{f^*(v)+d}{w_v})^2< w_u(\frac{f^*(u)}{w_u})^2+w_v(\frac{f^*(v)}{w_v})^2 \\
        &\Leftrightarrow \frac{(f^*(u)-d)^2}{w_u}+\frac{(f^*(v)+d)^2}{w_v}<\frac{f^*(u)^2}{w_u}+\frac{f^*(v)^2}{w_v} \\
        &\Leftrightarrow \frac{-2df^*(u) + d^2}{w_u} + \frac{2df^*(v) + d^2}{w_v} < 0 \\
        &\Leftrightarrow d^2(w_u+w_v) + 2d(w_uf^*(v) - w_vf^*(u)) < 0 \\
        &\Leftrightarrow d(w_u+w_v) < 2(w_vf^*(u) - w_uf^*(v))
    \end{split}
    \end{equation}
    Since $w_v f^*(u) > w_u f^*(v)$, the right-hand side (RHS) is a
    positive constant. As $d \to 0^+$, the left-hand side (LHS)
    approaches 0. Thus, a sufficiently small $d>0$ always exists
    that satisfies the inequality. This contradicts the optimality.
\end{proof}

\begin{theorem}
For any node $u \in B_i \setminus B_{i-1}$, let
    \begin{equation*}
    \lambda_u=\lambda_i\coloneqq{\frac{\WE(B_i)-\WE(B_{i-1})}{\WV(B_i \setminus B_{i-1})}},
    \end{equation*}
    the optimal solution $\bm{\ell}^*$ is unique and is given by ${\ell_v}^*=\lambda_v$ for each node $v$. And $\lambda_1>\lambda_2>...>\lambda_\gamma$.
    \label{prop:connection}
\end{theorem}

\begin{proof}[Proof of Theorem 3]
    Suppose $\bm{\ell}^*$ is one optimal solution, we sort them in descending order by $\ell^{*}_{u_1}\ge\ell^{*}_{u_2}\ge...\ge\ell^{*}_{u_{|\mathcal{V}|}}$.Suppose the first inequality is the $t$-th inequality, then 
    \[
    \ell^{*}_{u_1}=\ell^{*}_{u_2}=...=\ell^{*}_{u_t} \Rightarrow 
        \frac{f^*(u_1)}{w_{u_1}}=\frac{f^*(u_2)}{w_{u_2}}=...=\frac{f^*(u_t)}{w_{u_t}}
    \]
    Set $S:=\{u_1,u_2,...,u_t\}$, we claim $S=B_1$. Since $B_{0}=\emptyset$,
    \begin{equation} 
    \label{ieq:6}
    \begin{split}
        &\frac{\WE(S)-\WE(B_{0})}{\WV(S)-\WV(B_{0})}=\frac{\WE(S)}{\WV(S)} \\
        &\le \frac{f^*(u_1)+f^*(u_2)+...+f^*(u_t)}{w_{u_1}+w_{u_2}+...+w_{u_t}}=\ell^{*}_{u_1}=\ell^{*}_{u_2}=...=\ell^{*}_{u_t}
    \end{split}
    \end{equation}
    The first inequality holds because, on one hand, $\WV(S) = w_{u_1} + w_{u_2} + \dots + w_{u_t}$, and on the other hand, the weights of all hyperedges with endpoints entirely in $S$ are totally distributed to $f^*(u_1), \dots, f^*(u_t)$, which implies $\WE(S) \le f^*(u_1) + f^*(u_2) + \dots + f^*(u_t)$.
    The inequality is for the general case, but due to Theorem \ref{prop:one-way}, equality holds here. This is because any hyperedge $e$ with endpoints not all in $S$ does not assign a non-zero value to $f^*(u_1), \dots, f^*(u_t)$, i.e., $f_e^*(u) = 0$ for any $u \in \{u_1, \dots, u_t\}$.
    
    Therefore, we assert that $S = B_1$. This is because, by the inequality, no other subgraph (that is not a proper subset of $S$) has a ratio of the sum of edge weights to the sum of vertex weights exceeding
    $$
    \frac{f^*(u_1) + f^*(u_2) + \dots + f^*(u_t)}{w_{u_1} + w_{u_2} + \dots + w_{u_t}}.
    $$
    On the other hand, the equality
    $$
    \frac{\WE(S)}{\WV(S)} = \frac{f^*(u_1) + f^*(u_2) + \dots + f^*(u_t)}{w_{u_1} + w_{u_2} + \dots + w_{u_t}}
    $$
    holds, so $S = B_1$.
    Furthermore, by Theorem \ref{prop:uniquity1}, $B_1$ is unique, so this set $\{u_1, \dots, u_t\}$ is unique and corresponds exactly one-to-one with the nodes of $B_1$.
    For the nodes in $B_2 \setminus B_1, \dots, B_\gamma \setminus B_{\gamma-1}$, we can use mathematical induction to prove the result. Therefore, Theorem \ref{prop:connection} holds.
\end{proof}

\begin{algorithm}[!t]
\caption{BCD for Load Balancing via Water-Filling}
\label{alg:bcd_waterfill}
\begin{algorithmic}[1]
    \State \textbf{Input:} Hypergraph $G=(\mathcal{V}, \mathcal{E})$, vertex weights $\{w_v\}_{v \in \mathcal{V}}$, hyperedge weights $\{w_e\}_{e \in \mathcal{E}}$
    \State \textbf{Output:} The final allocation $f = \{f_e(v) \mid e \in \mathcal{E}, v \in e\}$
    \Statex
    \State \textbf{Initialize:} Create a feasible allocation $f$
    \For{each hyperedge $e \in \mathcal{E}$}
        \State $\WV(e) \gets \sum_{v \in e} w_v$ \Comment{Total vertex weight in $e$}
        \For{each vertex $v \in e$}
            \State $f_e(v) \gets w_e \cdot (w_v / \WV(e))$ \Comment{Proportional allocation}
        \EndFor
    \EndFor
    \Statex
    \State \textbf{repeat}
        \State $f_{\text{prev}} \gets f$ \Comment{Store allocation from previous iteration}
        \For{each hyperedge $e \in \mathcal{E}$ (in a fixed order)}
            \State \Comment{Locally optimize the allocation for hyperedge $e$}
            \State $ \{f_e(v)\}_{v \in e} \gets \WaterFill(e, w_e, \{w_v\}_{v \in e}, \{f_{e'}\}_{e' \ne e}) $
        \EndFor
    \State \textbf{until} convergence \Comment{e.g., $f = f_{\text{prev}}$ or $||f - f_{\text{prev}}|| < \epsilon$}
    \Statex
    \State \textbf{return} $f$
\end{algorithmic}
\end{algorithm}

Theorems~\ref{prop:uniquity1} and~\ref{prop:connection} establish the uniqueness and quadratic formulation of the decomposition. 
Crucially, Theorem~\ref{prop:one-way} implies that optimal allocation occurs only when nodes receiving weight have the minimum load. 
This motivates our Block-Coordinate Descent (BCD) method (Algorithm~\ref{alg:bcd_waterfill}), which employs the \textbf{water-filling} algorithm~\cite{boyd2004convex} to drive each hyperedge to its local optimum. 
Initialized with proportional allocations, the algorithm iteratively updates hyperedge weights via water-filling until convergence. 
Convergence to the global optimum of \eqref{eq:decomposition_hypergraph} is guaranteed by two key properties:
\begin{enumerate}[wide]
    \item The objective function,
    $$ \sum_{v \in \mathcal{V}} w_v \ell_v^2 =
       \sum_{v \in \mathcal{V}} w_v \left(
       \frac{\sum_{e \ni v} f_e(v)}{w_v}
       \right)^2 =
       \sum_{v \in \mathcal{V}} \frac{1}{w_v} \left(
       \sum_{e \ni v} f_e(v)
       \right)^2
    $$
    is continuously differentiable, as it is a quadratic
    polynomial of the variables $f_e(v)$.

   \item The objective is \textbf{blockwise strictly convex}.
    When optimizing for a single hyperedge $e$ (a block),
    the sub-problem's objective (ignoring constants) is:
    $$ \sum_{v \in e} \frac{1}{w_v} (f_e(v) + C_v)^2 \quad
       \text{where } C_v = \sum_{e' \ni v, e' \ne e} f_{e'}(v)
       \text{ is constant.}
    $$
    This function is strictly convex, as its Hessian matrix $H$
    is positive definite:
    $$ H = \text{diag}(2/w_{v_1}, 2/w_{v_2}, \dots, 2/w_{v_k}) $$
\end{enumerate}

As mentioned earlier, the key advantage of this technique lies in its ability to accurately characterize the density hierarchy of the hypergraph, such that subgraphs with different density levels are distinctly separated (assigned to different density layers $B_i$ and $B_j$). Since the coarsening process typically merges locally dense structures into single nodes, the locally-dense decomposition provides a global perspective and complementary guidance for the coarsest hypergraph in the initial partitioning phase. We regard Algorithm \ref{alg:bcd_waterfill} as a community detection method,
and employ the two-stage procedure in Section \ref{subsec:framework}
to improve our initial solution.

\section{experiments}
    \label{sec:exp}
    \begin{table*}[t]
\centering
\caption{Experimental results on Titan23 benchmark.}
\vspace{-.15in}
\label{tab:Titan23}
\begingroup 
\setlength{\tabcolsep}{2.5pt}
\resizebox{0.86\textwidth}{!}{
\begin{NiceTabular}{c||c|c|c|c|c||c|c|c|c||c|c|c|c} \toprule
\multirow{3}{*}{ Design }       &
       \multicolumn{5}{c}{$k=2$} &
       \multicolumn{4}{c}{$k=3$} &
       \multicolumn{4}{c}{$k=4$}
       \\ \cline{2-14}
\textbf{\begin{tabular}[c]{@{}c@{}} \phantom{1} \\ \phantom{2}\end{tabular}}&  \hspace{-1.5pt}\cellcolor[HTML]{F2F2F2}\textbf{hMETIS}\hspace{-1.5pt} &\cellcolor[HTML]{F2F2F2}\textbf{K-Spec}  & \hspace{-1.5pt}\cellcolor[HTML]{F2F2F2}\textbf{KaHyPar}\hspace{-1.5pt}  & \hspace{-4.5pt} \cellcolor[HTML]{F2F2F2}\textbf{SHyPar} \hspace{-4.5pt}   & \hspace{1pt} \cellcolor[HTML]{FFF2CC}\textbf{Ours} \hspace{1pt}   & \hspace{-1.5pt} \cellcolor[HTML]{F2F2F2}\textbf{hMETIS} \hspace{-1.5pt}   & \cellcolor[HTML]{F2F2F2}\textbf{K-Spec}   & \hspace{-4.5pt} \cellcolor[HTML]{F2F2F2}\textbf{KaHyPar} \hspace{-4.5pt}    & \hspace{1pt} \cellcolor[HTML]{FFF2CC}\textbf{Ours} \hspace{1pt}  &   \hspace{-1.5pt} \cellcolor[HTML]{F2F2F2}\textbf{hMETIS}  \hspace{-1.5pt}  & \cellcolor[HTML]{F2F2F2}\textbf{K-Spec}   & \hspace{-4.5pt} \cellcolor[HTML]{F2F2F2}\textbf{KaHyPar} \hspace{-4.5pt}     & \hspace{1pt} \cellcolor[HTML]{FFF2CC}\textbf{Ours} \hspace{1pt} \\
\hline \hline
sparcT1\_core  & 985                                  & 977                                  & {\color[HTML]{00B050} \textbf{974}}  & {\color[HTML]{00B050} \textbf{974}}  & {\color[HTML]{00B050} \textbf{974}}  & 2112 & 1889                                & 1765                                 & {\color[HTML]{00B050} \textbf{1682}} & 2617                                 & 2492                                 & 2591                                 & {\color[HTML]{00B050} \textbf{2111}} \\
neuron        & 284                                  & 244                                  & 245                                  & {\color[HTML]{00B050} \textbf{243}}  & 244                                  & 514  & {\color[HTML]{00B050} \textbf{396}} & 407                                  & 400                                  & 566                                  & 431                                  & 415                                  & {\color[HTML]{00B050} \textbf{384}}  \\
stereo\_vision & 173                                  & {\color[HTML]{00B050} \textbf{169}}  & {\color[HTML]{00B050} \textbf{169}}  & {\color[HTML]{00B050} \textbf{169}}  & {\color[HTML]{00B050} \textbf{169}}  & 342  & 336                                 & 352                                  & {\color[HTML]{00B050} \textbf{322}}  & 435                                  & 475                                  & 392                                  & {\color[HTML]{00B050} \textbf{386}}  \\
des90         & 380                                  & {\color[HTML]{00B050} \textbf{374}}  & 450                                  & 379                                  & 376                                  & 514  & 535                                 & 504                                  & {\color[HTML]{00B050} \textbf{503}}  & 747                                  & 747                                  & 691                                  & {\color[HTML]{00B050} \textbf{660}}  \\
SLAM\_spheric  & {\color[HTML]{00B050} \textbf{1061}} & {\color[HTML]{00B050} \textbf{1061}} & {\color[HTML]{00B050} \textbf{1061}} & {\color[HTML]{00B050} \textbf{1061}} & {\color[HTML]{00B050} \textbf{1061}} & 2854 & 2720                                & {\color[HTML]{00B050} \textbf{2674}} & 2686                                 & 3351                                 & 3241                                 & 3206                                 & {\color[HTML]{00B050} \textbf{3184}} \\
cholesky\_mc   & 287                                  & {\color[HTML]{00B050} \textbf{282}}  & 286                                  & 283                                  & 283                                  & 918  & {\color[HTML]{00B050} \textbf{864}} & 879                                  & 877                                  & 980                                  & 984                                  & 975                                  & {\color[HTML]{00B050} \textbf{975}}  \\
segmentation  & {\color[HTML]{00B050} \textbf{107}}  & 120                                  & {\color[HTML]{00B050} \textbf{107}}  & {\color[HTML]{00B050} \textbf{107}}  & {\color[HTML]{00B050} \textbf{107}}  & 476  & 453                                 & 440                                  & {\color[HTML]{00B050} \textbf{411}}  & 564                                  & 490                                  & 478                                  & {\color[HTML]{00B050} \textbf{477}}  \\
bitonic\_mesh  & 586                                  & 584                                  & 589                                  & 586                                  & {\color[HTML]{00B050} \textbf{583}}  & 927  & 895                                 & 893                                  & {\color[HTML]{00B050} \textbf{890}}  & 1149                                 & 1311                                 & 1102                                 & {\color[HTML]{00B050} \textbf{1097}} \\
dart          & 807                                  & 805                                  & 798                                  & {\color[HTML]{00B050} \textbf{784}}  & {\color[HTML]{00B050} \textbf{784}}  & 1258 & 1243                                & 1166                                 & {\color[HTML]{00B050} \textbf{1053}} & 1499                                 & 1401                                 & 1280                                 & {\color[HTML]{00B050} \textbf{1264}} \\
openCV        & 465                                  & {\color[HTML]{00B050} \textbf{434}}  & 686                                  & 499                                  & {\color[HTML]{00B050} \textbf{434}}  & 522  & 525                                 & 677                                  & {\color[HTML]{00B050} \textbf{489}}  & 581                                  & 522                                  & 693                                  & {\color[HTML]{00B050} \textbf{520}}  \\
stap\_qrd      & 383                                  & 464                                  & {\color[HTML]{00B050} \textbf{371}}  & {\color[HTML]{00B050} \textbf{371}}  & {\color[HTML]{00B050} \textbf{371}}  & 533  & 497                                 & 496                                  & {\color[HTML]{00B050} \textbf{468}}  & 760                                  & 674                                  & 614                                  & {\color[HTML]{00B050} \textbf{611}}  \\
minres        & 211                                  & {\color[HTML]{00B050} \textbf{207}}  & {\color[HTML]{00B050} \textbf{207}}  & {\color[HTML]{00B050} \textbf{207}}  & {\color[HTML]{00B050} \textbf{207}}  & 343  & {\color[HTML]{00B050} \textbf{309}} & {\color[HTML]{00B050} \textbf{309}}  & {\color[HTML]{00B050} \textbf{309}}  & 415                                  & {\color[HTML]{00B050} \textbf{407}}  & 443                                  & 410                                  \\
cholesky\_bdti & 1172                                 & {\color[HTML]{00B050} \textbf{1136}} & 1238                                 & 1156                                 & 1156                                 & 1929 & 1755                                & 1716                                 & {\color[HTML]{00B050} \textbf{1652}} & 1881                                 & {\color[HTML]{00B050} \textbf{1865}} & 1875                                 & 1874                                 \\
denoise       & 456                                  & 418                                  & {\color[HTML]{00B050} \textbf{416}}  & {\color[HTML]{00B050} \textbf{416}}  & {\color[HTML]{00B050} \textbf{416}}  & 925  & 915                                 & 782                                  & {\color[HTML]{00B050} \textbf{721}}  & 1015                                 & 1001                                 & 890                                  & {\color[HTML]{00B050} \textbf{862}}  \\
sparcT2\_core  & 1214                                 & 1188                                 & {\color[HTML]{00B050} \textbf{1183}} & {\color[HTML]{00B050} \textbf{1183}} & {\color[HTML]{00B050} \textbf{1183}} & 3172 & 2249                                & 2057                                 & {\color[HTML]{00B050} \textbf{2032}} & 3088                                 & 3558                                 & {\color[HTML]{00B050} \textbf{2741}} & 2752                                 \\
gsm\_switch    & 4881                                 & 1833                                 & 1651                                 & {\color[HTML]{00B050} \textbf{1621}} & 1645                                 & 5213 & 3694                                & {\color[HTML]{00B050} \textbf{2432}} & 2439                                 & 5379                                 & 4404                                 & {\color[HTML]{00B050} \textbf{2855}} & 2864                                 \\
mes\_noc       & 671                                  & {\color[HTML]{00B050} \textbf{633}}  & 649                                  & 651                                  & 640                                  & 1158 & 1125                                & {\color[HTML]{00B050} \textbf{1110}} & 1123                                 & 1459                                 & 1346                                 & {\color[HTML]{00B050} \textbf{1307}} & {\color[HTML]{00B050} \textbf{1307}} \\
LU230         & 3523                                 & {\color[HTML]{00B050} \textbf{3363}} & 3555                                 & 3602                                 & 3529                                 & 4677 & 4548                                & 5564                                 & {\color[HTML]{00B050} \textbf{5331}} & {\color[HTML]{00B050} \textbf{5748}} & 6310                                 & 5915                                 & 5958                                 \\
LU\_Network    & 524                                  & 524                                  & 524                                  & 524                                  & {\color[HTML]{00B050} \textbf{523}}  & 967  & 882                                 & {\color[HTML]{00B050} \textbf{784}}  & {\color[HTML]{00B050} \textbf{784}}  & 1507                                 & 1417                                 & 1278                                 & {\color[HTML]{00B050} \textbf{1248}} \\
sparcT1\_chip2 & 908                                  & 876                                  & 874                                  & {\color[HTML]{00B050} \textbf{873}}  & 874                                  & 1490 & 1404                                & 1256                                 & {\color[HTML]{00B050} \textbf{1177}} & 2151                                 & 1601                                 & 1595                                 & {\color[HTML]{00B050} \textbf{1431}} \\
directrf      & 538                                  & {\color[HTML]{00B050} \textbf{515}}  & 631                                  & 632                                  & 630                                  & 730  & 762                                 & 1013                                 & {\color[HTML]{00B050} \textbf{710}}  & 1071                                 & 1092                                 & 1069                                 & {\color[HTML]{00B050} \textbf{974}}  \\
bitcoin\_miner & 1574                                 & 1562                                 & 1541                                 & 1514                                 & {\color[HTML]{00B050} \textbf{1489}} & 2119 & 1917                                & 1886                                 & {\color[HTML]{00B050} \textbf{1867}} & 2800                                 & 2737                                 & 1894                                 & {\color[HTML]{00B050} \textbf{1834}}\\ \hline
\textbf{\begin{tabular}[c]{@{}c@{}} Norm \\ Avg.\end{tabular}}    &  \cellcolor[HTML]{F2F2F2}1.035  & \cellcolor[HTML]{F2F2F2}0.975       & \cellcolor[HTML]{F2F2F2}1 & \cellcolor[HTML]{F2F2F2}0.972    & \cellcolor[HTML]{FFF2CC}\textbf{0.964} & \cellcolor[HTML]{F2F2F2}1.102  & \cellcolor[HTML]{F2F2F2}1.016& \cellcolor[HTML]{F2F2F2}1 &  \cellcolor[HTML]{FFF2CC}\textbf{0.943}     & \cellcolor[HTML]{F2F2F2}1.131& \cellcolor[HTML]{F2F2F2}1.080        & \cellcolor[HTML]{F2F2F2}1   & \cellcolor[HTML]{FFF2CC}\textbf{0.955}        \\ \bottomrule

\end{NiceTabular}
}
\endgroup
\par\vspace{2pt}
{\raggedright \footnotesize
$^\star$ The Normalized Average (Norm. Avg.) for each method is calculated as the geometric mean relative to KaHyPar.
\par}
\vspace{-.1in}
\end{table*}

\begin{table*}[!t]
\centering
\caption{Experimental results on ISPD98 benchmark.}
\vspace{-.1in}
\label{tab:ISPD98}
\begingroup 
\setlength{\tabcolsep}{2.5pt}
\resizebox{0.86\textwidth}{!}{

\begin{NiceTabular}{c||c|c|c|c|c||c|c|c|c|c||c|c|c|c|c} \toprule
\multirow{3}{*}{ Design }       &
       \multicolumn{5}{c}{$k=2$} &
       \multicolumn{5}{c}{$k=3$} &
       \multicolumn{5}{c}{$k=4$}
       \\ \cline{2-16}
\textbf{\begin{tabular}[c]{@{}c@{}} \phantom{1} \\ \phantom{2}\end{tabular}}&  \hspace{-1.5pt}\cellcolor[HTML]{F2F2F2}\textbf{hMETIS} \hspace{-1.5pt} &\cellcolor[HTML]{F2F2F2}\textbf{K-Spec}  & \hspace{-1.5pt}\cellcolor[HTML]{F2F2F2}\textbf{KaHyPar}\hspace{-1.5pt}  & \hspace{-4.5pt} \cellcolor[HTML]{F2F2F2}\textbf{SHyPar} \hspace{-4.5pt}   & \hspace{1pt} \cellcolor[HTML]{FFF2CC}\textbf{Ours} \hspace{1pt}   & \hspace{-1.5pt} \cellcolor[HTML]{F2F2F2}\textbf{hMETIS} \hspace{-1.5pt}   & \cellcolor[HTML]{F2F2F2}\textbf{K-Spec}   & \hspace{-4.5pt} \cellcolor[HTML]{F2F2F2}\textbf{KaHyPar} \hspace{-4.5pt}  & \hspace{-4.5pt} \cellcolor[HTML]{F2F2F2}\textbf{SHyPar} \hspace{-4.5pt}  & \hspace{1pt} \cellcolor[HTML]{FFF2CC}\textbf{Ours} \hspace{1pt}  &  \hspace{-1.5pt}\cellcolor[HTML]{F2F2F2}\textbf{hMETIS}  \hspace{-1.5pt}  & \cellcolor[HTML]{F2F2F2}\textbf{K-Spec}   & \hspace{-4.5pt} \cellcolor[HTML]{F2F2F2}\textbf{KaHyPar} \hspace{-4.5pt}  & \hspace{-4.5pt} \cellcolor[HTML]{F2F2F2}\textbf{SHyPar} \hspace{-4.5pt}   & \hspace{1pt} \cellcolor[HTML]{FFF2CC}\textbf{Ours} \hspace{1pt} \\
\hline \hline
ibm01 & 205  & 203                                  & 202                                  & {\color[HTML]{00B050} \textbf{201}}  & 202                                  & 352  & 352                                 & {\color[HTML]{00B050} \textbf{335}}  & 338  & {\color[HTML]{00B050} \textbf{335}}  & 485  & 522                                  & 463                                  & 466                                 & {\color[HTML]{00B050} \textbf{457}}  \\
ibm02 & 345  & 333                                  & 329                                  & {\color[HTML]{00B050} \textbf{327}}  & {\color[HTML]{00B050} \textbf{327}}  & 340  & {\color[HTML]{00B050} \textbf{339}} & 340                                  & 340  & {\color[HTML]{00B050} \textbf{339}}  & 609  & 706                                  & 642                                  & {\color[HTML]{00B050} \textbf{587}} & 602                                  \\
ibm03 & 961  & 957                                  & 957                                  & {\color[HTML]{00B050} \textbf{952}}  & {\color[HTML]{00B050} \textbf{952}}  & 1483 & 1480                                & 1627                                 & 1549 & {\color[HTML]{00B050} \textbf{1434}} & 1673 & 1690                                 & {\color[HTML]{00B050} \textbf{1650}} & 1718                                & 1655                                 \\
ibm04 & 581  & 580                                  & 580                                  & {\color[HTML]{00B050} \textbf{579}}  & 580                                  & 1176 & 1212                                & 1160                                 & 1146 & {\color[HTML]{00B050} \textbf{1135}} & 1638 & 1626                                 & 1590                                 & 1683                                & {\color[HTML]{00B050} \textbf{1589}} \\
ibm05 & 1728 & 1716                                 & {\color[HTML]{00B050} \textbf{1706}} & 1707                                 & {\color[HTML]{00B050} \textbf{1706}} & 2666 & 2635                                & 2749                                 & 2764 & {\color[HTML]{00B050} \textbf{2619}} & 2950 & {\color[HTML]{00B050} \textbf{2946}} & 2956                                 & 3097                                & 2956                                 \\
ibm06 & 987  & 976                                  & 980                                  & 969                                  & {\color[HTML]{00B050} \textbf{963}}  & 1293 & 1305                                & 1297                                 & 1323 & {\color[HTML]{00B050} \textbf{1280}} & 1493 & 1476                                 & {\color[HTML]{00B050} \textbf{1470}} & 1522                                & {\color[HTML]{00B050} \textbf{1470}} \\
ibm07 & 926  & 935                                  & 900                                  & {\color[HTML]{00B050} \textbf{882}}  & 888                                  & 1852 & 1846                                & 1774                                 & 1722 & {\color[HTML]{00B050} \textbf{1697}} & 2206 & 2154                                 & 2066                                 & 2076                                & {\color[HTML]{00B050} \textbf{2062}} \\
ibm08 & 1147 & {\color[HTML]{00B050} \textbf{1140}} & {\color[HTML]{00B050} \textbf{1140}} & {\color[HTML]{00B050} \textbf{1140}} & {\color[HTML]{00B050} \textbf{1140}} & 2000 & 2037                                & 2050                                 & 1944 & {\color[HTML]{00B050} \textbf{1899}} & 2332 & 2328                                 & 2291                                 & 2340                                & {\color[HTML]{00B050} \textbf{2246}} \\
ibm09 & 628  & {\color[HTML]{00B050} \textbf{620}}  & {\color[HTML]{00B050} \textbf{620}}  & {\color[HTML]{00B050} \textbf{620}}  & {\color[HTML]{00B050} \textbf{620}}  & 1395 & 1384                                & 1427                                 & 1357 & {\color[HTML]{00B050} \textbf{1322}} & 1700 & 1676                                 & {\color[HTML]{00B050} \textbf{1666}} & 1701                                & 1667                                 \\
ibm10 & 1335 & 1257                                 & 1313                                 & {\color[HTML]{00B050} \textbf{1254}} & {\color[HTML]{00B050} \textbf{1254}} & 1924 & 1880                                & 1873                                 & 1913 & {\color[HTML]{00B050} \textbf{1872}} & 2444 & 2400                                 & 2190                                 & 2234                                & {\color[HTML]{00B050} \textbf{2187}} \\
ibm11 & 1062 & {\color[HTML]{00B050} \textbf{1051}} & 1062                                 & {\color[HTML]{00B050} \textbf{1051}} & 1062                                 & 1810 & 1843                                & 1777                                 & 1784 & {\color[HTML]{00B050} \textbf{1739}} & 2432 & 2452                                 & 2383                                 & 2472                                & {\color[HTML]{00B050} \textbf{2346}} \\
ibm12 & 1952 & 1937                                 & {\color[HTML]{00B050} \textbf{1920}} & 1986                                 & {\color[HTML]{00B050} \textbf{1920}} & 2874 & 2791                                & 2778                                 & 2778 & {\color[HTML]{00B050} \textbf{2718}} & 3893 & 3844                                 & {\color[HTML]{00B050} \textbf{3774}} & 3802                                & 3776                                 \\
ibm13 & 833  & 832                                  & 848                                  & {\color[HTML]{00B050} \textbf{831}}  & {\color[HTML]{00B050} \textbf{831}}  & 1390 & 1335                                & 1386                                 & 1312 & {\color[HTML]{00B050} \textbf{1272}} & 1865 & 1904                                 & {\color[HTML]{00B050} \textbf{1715}} & 1736                                & {\color[HTML]{00B050} \textbf{1715}} \\
ibm14 & 1892 & 1850                                 & 1849                                 & {\color[HTML]{00B050} \textbf{1842}} & 1849                                 & 2737 & 2710                                & 2606                                 & 2691 & {\color[HTML]{00B050} \textbf{2603}} & 3283 & 3475                                 & {\color[HTML]{00B050} \textbf{3192}} & 3230                                & 3195                                 \\
ibm15 & 2751 & 2741                                 & {\color[HTML]{00B050} \textbf{2728}} & {\color[HTML]{00B050} \textbf{2728}} & {\color[HTML]{00B050} \textbf{2728}} & 4121 & 4333                                & 4132                                 & 4140 & {\color[HTML]{00B050} \textbf{4008}} & 4825 & 4720                                 & 4724                                 & 4809                                & {\color[HTML]{00B050} \textbf{4584}} \\
ibm16 & 1960 & 1921                                 & 1929                                 & 1887                                 & {\color[HTML]{00B050} \textbf{1855}} & 3120 & 3062                                & 3477                                 & 3043 & {\color[HTML]{00B050} \textbf{2890}} & 4182 & 4060                                 & 3805                                 & 3802                                & {\color[HTML]{00B050} \textbf{3757}} \\
ibm17 & 2336 & 2307                                 & 2294                                 & {\color[HTML]{00B050} \textbf{2285}} & {\color[HTML]{00B050} \textbf{2285}} & 4352 & 4248                                & {\color[HTML]{00B050} \textbf{4053}} & 4095 & 4054                                 & 5708 & 5583                                 & 5494                                 & 5386                                & {\color[HTML]{00B050} \textbf{5213}} \\
ibm18 & 1827 & 1523                                 & 1915                                 & {\color[HTML]{00B050} \textbf{1521}} & {\color[HTML]{00B050} \textbf{1521}} & 2706 & 2401                                & {\color[HTML]{00B050} \textbf{2331}} & 2361 & 2334                                 & 3262 & 2918                                 & {\color[HTML]{00B050} \textbf{2822}} & 2993                                & {\color[HTML]{00B050} \textbf{2822}}           \\ \hline
\textbf{\begin{tabular}[c]{@{}c@{}} Norm \\ Avg.\end{tabular}}    &  \cellcolor[HTML]{F2F2F2}1.009  & \cellcolor[HTML]{F2F2F2}0.987       & \cellcolor[HTML]{F2F2F2}1 & \cellcolor[HTML]{F2F2F2}0.981    & \cellcolor[HTML]{FFF2CC}\textbf{0.979} & \cellcolor[HTML]{F2F2F2}1.010  & \cellcolor[HTML]{F2F2F2}0.994& \cellcolor[HTML]{F2F2F2}1 & \cellcolor[HTML]{F2F2F2}0.985 & \cellcolor[HTML]{FFF2CC}\textbf{0.959}     & \cellcolor[HTML]{F2F2F2}1.041& \cellcolor[HTML]{F2F2F2}1.058        & \cellcolor[HTML]{F2F2F2}1 & \cellcolor[HTML]{F2F2F2}1.016   & \cellcolor[HTML]{FFF2CC}\textbf{0.989}        \\ \bottomrule

\end{NiceTabular}
}
\endgroup
\vspace{-.1in}
\end{table*}

The structure of this chapter is organized as follows. In Section~\ref{subsec:setup}, we introduce the experimental setup. In Section~\ref{subsec:performance}, we compare our algorithm with the leading methods, including \textit{hMETIS}, \textit{KaHyPar}, \textit{SpecPart}, and \textit{SHyPar}, on both the Titan23 and ISPD98 benchmarks. %In Subsection~\ref{subsec:ablation}, we present the ablation study, which analyzes how various community detection methods in the uncoarsening phase and the locally-dense decomposition in the initial partitioning phase contribute to performance improvement. By comparing different community detection methods, we summarize general principles for selecting suitable methods across different benchmarks.

\subsection{Experimental setup}
\label{subsec:setup}

\begin{enumerate}
    \item  In Sections~\ref{exp:titan23} and~\ref{exp:ispd98}, we evaluate ComPart on the Titan23 and ISPD98 benchmarks against leading algorithms, including \textit{hMETIS}, \textit{KaHyPar}, \textit{SpecPart}, and \textit{SHyPar}. We employ four iterations ($\alpha=4$) of the \texttt{Fast Greedy} algorithm for community detection. To ensure fair comparison, ComPart, \textit{hMETIS}, and \textit{KaHyPar} were restricted to \textbf{the same total execution time}, reporting the best result found. Results for \textit{SpecPart} and \textit{SHyPar} ($k=2$) are cited from their respective papers. For \textit{SHyPar} with $k>2$, we executed the source code but limited the comparison to ISPD98 due to excessive runtime on Titan23 (over 20 hours for medium-sized datasets). All experiments enforce a maximum imbalance of 2\% per partition (i.e., $\epsilon = 0.02 k$).
    \item In Section \ref{exp:ablation}, we perform an ablation study on the two components of our method: community detection in uncoarsening and locally-dense decomposition, and compare the results with the performance of KaHyPar and ShYPar.
    \item In Section \ref{exp:cd}, we compare the performance of different community detection strategies applied during the uncoarsening phase, and also compare our results with KaHyPar and ShYPar.
\end{enumerate}

Our implementation is based on \texttt{KaHyPar}, using the configuration file \emph{cut\_kKaHyPar\_sea20.ini}. All experiments were executed on a single core of an \texttt{Intel(R) Xeon(R) Gold 6438Y+ CPU @ 2.00GHz} system running \texttt{Ubuntu 22.04.5 LTS} with \texttt{512GB RAM}. The implementation was compiled from source using \texttt{g++ 11.4.0} under the \texttt{C++17} standard.

\subsection{Comprehensive performance on standard benchmarks}
\label{subsec:performance}

\subsubsection{ComPart against Baselines on Titan23 Benchmark}
\label{exp:titan23}
Table \ref{tab:Titan23} compares our method, ComPart, with the baselines on the Titan23 benchmark, indicating that our algorithm outperforms all existing algorithms on all setups. The Titan23 benchmark features many large hyperedges (in terms of the number of nodes). Overall, the best baseline is SHyPar. When $k=2$, we improve by 3.62\% over the base partitioner KaHyPar and 0.87\% over the best baseline (SHyPar). When $k=3$, we improve by 5.74\% over KaHyPar. When $k=4$, we improve by 4.50\% over KaHyPar. In general, our lead is larger for $k=3$ and $k=4$ than for $k=2$. This suggests that existing works are generally stronger in the bipartition domain ($k=2$) but weaker for $k>2$. Furthermore, our algorithm achieves a greater overall improvement on Titan23 than on ISPD98. We attribute this to our community detection strategy, which can effectively handle cases with large hyperedge sizes better than these baselines. 

\subsubsection{ComPart against Baselines on ISPD98 Benchmark}
\label{exp:ispd98}
Table \ref{tab:ISPD98} presents the comparison results of our method, ComPart, against the baselines on the ISPD98 benchmark. The hyperedge sizes in the ISPD98 benchmark are typically smaller. Overall, the best baseline is still SHyPar. Our lead is greater for $k=3$ and $k=4$ compared to $k=2$. When $k=2$, we achieve a 2.10\% lead over the base partitioner KaHyPar and a 0.17\% lead relative to SHyPar. When $k=3$, we achieve a 4.11\% lead over KaHyPar and a 2.68\% lead over SHyPar. When $k=4$, we achieve a 1.14\% lead over KaHyPar and a 2.70\% lead over SHyPar. 

\subsubsection{Ablation Study}
\label{exp:ablation}
Fig. \ref{fig:abl} exhibits the ablation study to evaluate the contributions of two key components of \texttt{ComPart}: Locally-Dense Decomposition in Initial Partitioning (I) and Community Detection in Uncoarsening (U).
We also compare our results with \texttt{KaHyPar} and \texttt{SHyPar}.
The results show that `U' is the primary contributor to the performance gain.
Specifically, the `U' component alone achieves a 3.6\% improvement over \texttt{KaHyPar}, while the `I' component alone provides a 0.7\% improvement.

\begin{figure}[!t]
  \centering
  \begin{minipage}[b]{0.48\linewidth} 
    \centering
    \includegraphics[width=\linewidth]{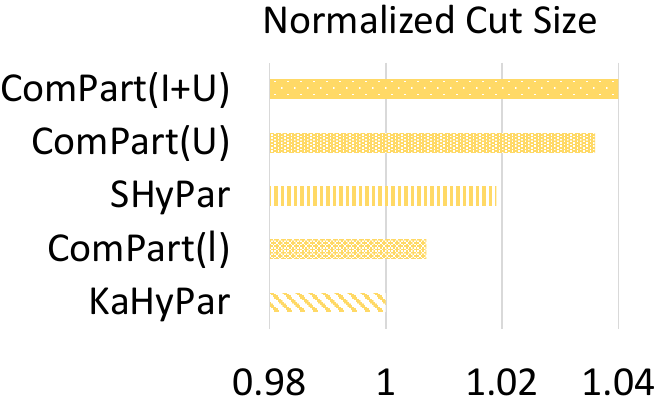}
    \vspace{-.1in}
    \caption{Ablation study.}
    \label{fig:abl}
  \end{minipage}
  \hfill 
  \begin{minipage}[b]{0.48\linewidth}
    \centering
    \includegraphics[width=\linewidth]{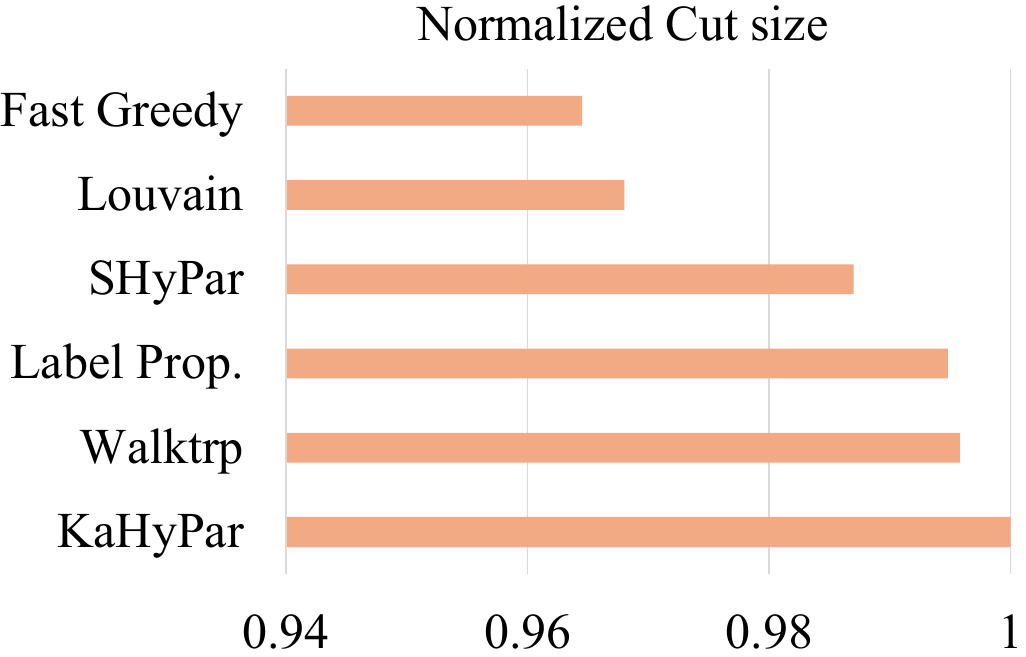}
    \vspace{-.3in}
    \caption{Different Community Detection Strategies.}
    \label{fig:cd}
  \end{minipage}
  \vspace{-.2in}
\end{figure}

\subsubsection{Comparison of Different Community Detection Strategies in the Uncoarsening Phase}
\label{exp:cd}
Figure \ref{fig:cd} presents the performance comparison of Fast Greedy, Louvain, Label Propagation, and Walktrap ($\alpha=4$ is used for each community detection strategy) within the uncoarsening phase. We include KaHyPar and SHyPar as strong baselines for reference. The results demonstrate that Fast Greedy and Louvain yield superior quality gains, whereas the contributions from Label Propagation and Walktrap are relatively modest.

\section{Conclusion}
    \label{sec:con}
    In this paper, we propose ComPart, a novel framework that integrates community detection into the uncoarsening phase to guide global solution exploration. 
Furthermore, theoretically generalize the \textit{locally-dense decomposition} technique from graphs to hypergraphs, providing a rigorous foundation for improving the initial partitioning. 
Experimental results demonstrate that our method consistently achieves state-of-the-art performance on standard benchmarks compared to existing solvers.

%\bibliographystyle{IEEEtran}
    %\bibliography{ref}

\section{Acknowledgement}
    This work is supported by Hong Kong Research Grants Council (RGC) CRF-YCRG C6003-24Y, GRF 16216825. It was partially conducted by ACCESS – AI Chip Center for Emerging Smart Systems, supported by the InnoHK initiative of the Innovation and Technology Commission of the Hong Kong Special Administrative Region Government.

\bibliographystyle{ACM-Reference-Format}
    \bibliography{ref}

@article{blondel2008fast,
  title={Fast unfolding of communities in large networks},
  author={Blondel, Vincent D and Guillaume, Jean-Loup and Lambiotte, Renaud and Lefebvre, Etienne},
  journal={Journal of statistical mechanics: theory and experiment},
  volume={2008},
  number={10},
  pages={P10008},
  year={2008},
  publisher={IOP Publishing}
}

@inproceedings{bustany2022specpart,
  title={SpecPart: A supervised spectral framework for hypergraph partitioning solution improvement},
  author={Bustany, Ismail and Kahng, Andrew B and Koutis, Ioannis and Pramanik, Bodhisatta and Wang, Zhiang},
  booktitle={Proceedings of the 41st IEEE/ACM International Conference on Computer-Aided Design},
  pages={1--9},
  year={2022}
}

@article{bustany2023k,
  title={K-SpecPart: Supervised embedding algorithms and cut overlay for improved hypergraph partitioning},
  author={Bustany, Ismail and Kahng, Andrew B and Koutis, Ioannis and Pramanik, Bodhisatta and Wang, Zhiang},
  journal={IEEE Transactions on Computer-Aided Design of Integrated Circuits and Systems},
  volume={43},
  number={4},
  pages={1232--1245},
  year={2023},
  publisher={IEEE}
}

@article{gottesburen2019evaluation,
  title={Evaluation of a flow-based hypergraph bipartitioning algorithm},
  author={Gottesb{\"u}ren, Lars and Hamann, Michael and Wagner, Dorothea},
  journal={arXiv preprint arXiv:1907.02053},
  year={2019}
}

@inproceedings{aghdaei2022hyperef,
  title={HyperEF: Spectral hypergraph coarsening by effective-resistance clustering},
  author={Aghdaei, Ali and Feng, Zhuo},
  booktitle={Proceedings of the 41st IEEE/ACM International Conference on Computer-Aided Design},
  pages={1--9},
  year={2022}
}

@inproceedings{liang2024medpart,
  title={Medpart: A multi-level evolutionary differentiable hypergraph partitioner},
  author={Liang, Rongjian and Agnesina, Anthony and Ren, Haoxing},
  booktitle={Proceedings of the 2024 International Symposium on Physical Design},
  pages={3--11},
  year={2024}
}

@misc{ccatalyurek2011patoh,
  title={PaToH (Partitioning Tool for Hypergraphs).},
  author={{\c{C}}ataly{\"u}rek, {\"U}mit V and Aykanat, Cevdet},
  year={2011}
}

@inproceedings{karypis1997multilevel,
  title={Multilevel hypergraph partitioning: Application in VLSI domain},
  author={Karypis, George and Aggarwal, Rajat and Kumar, Vipin and Shekhar, Shashi},
  booktitle={Proceedings of the 34th annual Design Automation Conference},
  pages={526--529},
  year={1997}
}

@article{schlag2023high,
  title={High-quality hypergraph partitioning},
  author={Schlag, Sebastian and Heuer, Tobias and Gottesb{\"u}ren, Lars and Akhremtsev, Yaroslav and Schulz, Christian and Sanders, Peter},
  journal={ACM Journal of Experimental Algorithmics},
  volume={27},
  pages={1--39},
  year={2023},
  publisher={ACM New York, NY}
}

@incollection{fiduccia1988linear,
  title={A linear-time heuristic for improving network partitions},
  author={Fiduccia, Charles M and Mattheyses, Robert M},
  booktitle={Papers on Twenty-five years of electronic design automation},
  pages={241--247},
  year={1988}
}

@article{kernighan1970efficient,
  title={An efficient heuristic procedure for partitioning graphs},
  author={Kernighan, Brian W and Lin, Shen},
  journal={The Bell system technical journal},
  volume={49},
  number={2},
  pages={291--307},
  year={1970},
  publisher={Nokia Bell Labs}
}

@article{sajadinia2025shypar,
  title={SHyPar: A Spectral Coarsening Approach to Hypergraph Partitioning},
  author={Sajadinia, Hamed and Aghdaei, Ali and Feng, Zhuo},
  journal={IEEE Transactions on Computer-Aided Design of Integrated Circuits and Systems},
  year={2025},
  publisher={IEEE}
}

@phdthesis{heuer2015engineering,
  title={Engineering initial partitioning algorithms for direct k-way hypergraph partitioning},
  author={Heuer, Tobias},
  year={2015},
  school={Karlsruher Institut f{\"u}r Technologie (KIT)}
}

@article{zhu2023fast,
  title={Fast Searching The Densest Subgraph And Decomposition With Local Optimality},
  author={Zhu, Yugao and Liu, Shenghua and Feng, Wenjie and Cheng, Xueqi},
  journal={arXiv preprint arXiv:2307.15969},
  year={2023}
}

@inproceedings{chen2024hypergraph,
  title={A Hypergraph Partitioner Utilizing a Novel Graph Generative Model},
  author={Chen, Magi and Wang, Ting-Chi},
  booktitle={Proceedings of the 43rd IEEE/ACM International Conference on Computer-Aided Design},
  pages={1--9},
  year={2024}
}

@inproceedings{tong2024easypart,
  title={EasyPart: An effective and comprehensive hypergraph partitioner for FPGA-based emulation},
  author={Tong, Shengbo and Li, Haoyuan and Xu, Jiahao and Pei, Chunyan and Yu, Wenjian and Liu, Shengjun and Shen, Jian},
  booktitle={Proceedings of the 43rd IEEE/ACM International Conference on Computer-Aided Design},
  pages={1--9},
  year={2024}
}

@article{li2024mapart,
  title={MaPart: An Efficient Multi-FPGA System-Aware Hypergraph Partitioning Framework},
  author={Li, Benzheng and Bi, Shunyang and You, Hailong and Qi, Zhongdong and Guo, Guangxin and Sun, Richard and Zhang, Yuming},
  journal={IEEE Transactions on Computer-Aided Design of Integrated Circuits and Systems},
  volume={43},
  number={10},
  pages={3212--3225},
  year={2024},
  publisher={IEEE}
}

@inproceedings{andre2018memetic,
  title={Memetic multilevel hypergraph partitioning},
  author={Andre, Robin and Schlag, Sebastian and Schulz, Christian},
  booktitle={Proceedings of the Genetic and Evolutionary Computation Conference},
  pages={347--354},
  year={2018}
}

@inproceedings{acikalin2022multilevel,
  title={Multilevel memetic hypergraph partitioning with greedy recombination},
  author={Acikalin, Utku Umur and Caskurlu, Bugra},
  booktitle={Proceedings of the Genetic and Evolutionary Computation Conference Companion},
  pages={168--171},
  year={2022}
}

@inproceedings{popp2021multilevel,
  title={Multilevel Acyclic Hypergraph Partitioning},
  author={Popp, Merten and Schlag, Sebastian and Schulz, Christian and Seemaier, Daniel},
  booktitle={2021 Proceedings of the Workshop on Algorithm Engineering and Experiments (ALENEX)},
  pages={1--15},
  year={2021},
  organization={SIAM}
}

@article{raghavan2007near,
  title={Near linear time algorithm to detect community structures in large-scale networks},
  author={Raghavan, Usha Nandini and Albert, R{\'e}ka and Kumara, Soundar},
  journal={Physical Review E—Statistical, Nonlinear, and Soft Matter Physics},
  volume={76},
  number={3},
  pages={036106},
  year={2007},
  publisher={APS}
}

@inproceedings{pons2005computing,
  title={Computing communities in large networks using random walks},
  author={Pons, Pascal and Latapy, Matthieu},
  booktitle={International symposium on computer and information sciences},
  pages={284--293},
  year={2005},
  organization={Springer}
}

@article{clauset2004finding,
  title={Finding community structure in very large networks},
  author={Clauset, Aaron and Newman, Mark EJ and Moore, Cristopher},
  journal={Physical Review E—Statistical, Nonlinear, and Soft Matter Physics},
  volume={70},
  number={6},
  pages={066111},
  year={2004},
  publisher={APS}
}

@inproceedings{danisch2017large,
  title={Large scale density-friendly graph decomposition via convex programming},
  author={Danisch, Maximilien and Chan, T-H Hubert and Sozio, Mauro},
  booktitle={Proceedings of the 26th International Conference on World Wide Web},
  pages={233--242},
  year={2017}
}

@book{boyd2004convex,
  title={Convex optimization},
  author={Boyd, Stephen and Vandenberghe, Lieven},
  year={2004},
  publisher={Cambridge university press}
}

@article{wolf2008multiprocessor,
  title={Multiprocessor system-on-chip (MPSoC) technology},
  author={Wolf, Wayne and Jerraya, Ahmed Amine and Martin, Grant},
  journal={IEEE transactions on computer-aided design of integrated circuits and systems},
  volume={27},
  number={10},
  pages={1701--1713},
  year={2008},
  publisher={IEEE}
}

@inproceedings{dick1998tgff,
  title={TGFF: task graphs for free},
  author={Dick, Robert P and Rhodes, David L and Wolf, Wayne},
  booktitle={Proceedings of the Sixth International Workshop on Hardware/Software Codesign.(CODES/CASHE'98)},
  pages={97--101},
  year={1998},
  organization={IEEE}
}

@book{graham2018kronecker,
  title={Kronecker products and matrix calculus with applications},
  author={Graham, Alexander},
  year={2018},
  publisher={Courier Dover Publications}
}

\end{document}